\documentclass[lineno]{biometrika}

\usepackage{amsmath}

\usepackage{times}
\usepackage{bm}
\usepackage{natbib}

\usepackage[plain,noend]{algorithm2e}

\makeatletter
\renewcommand{\algocf@captiontext}[2]{#1\algocf@typo. \AlCapFnt{}#2} 
\def\@algocf@capt@plain{top}
\renewcommand{\algocf@makecaption}[2]{%
  \addtolength{\hsize}{\algomargin}%
  \sbox\@tempboxa{\algocf@captiontext{#1}{#2}}%
  \ifdim\wd\@tempboxa >\hsize
    \hskip .5\algomargin%
    \parbox[t]{\hsize}{\algocf@captiontext{#1}{#2}}
  \else%
    \global\@minipagefalse%
    \hbox to\hsize{\box\@tempboxa}
  \fi%
  \addtolength{\hsize}{-\algomargin}%
}
\makeatother

\def\RD{\textsc{RD}}
\def\RR{\textsc{RR}}
\def\SRR{\textsc{SRR}}
\def\CRR{\textsc{CRR}}

\def\ind{\begin{picture}(9,8)
         \put(0,0){\line(1,0){9}}
         \put(3,0){\line(0,1){8}}
         \put(6,0){\line(0,1){8}}
         \end{picture}
        }

 \def\pr{\text{pr}}

\begin{document}

\jyear{2014}
\jvol{101}
\jnum{1}
\copyrightinfo{\Copyright\ 2014 Biometrika Trust\goodbreak {\em Printed in Great Britain}}


\markboth{Peng Ding \and Tyler J. VanderWeele}{Generalized Cornfield conditions}

\title{Generalized Cornfield conditions for the risk difference}

\author{Peng Ding}
\affil{Department of Statistics, Harvard University, One Oxford Street, Cambridge, Massachusetts, 02138, U.S.A.  \email{pengding@fas.harvard.edu}}

\author{\and Tyler J. VanderWeele}
\affil{Departments of Biostatistics and Epidemiology, Harvard School of Public Health, Boston,
Massachusetts 02115, U.S.A. \email{tvanderw@hsph.harvard.edu}}

\maketitle

\begin{abstract}
A central question in causal inference with observational studies is the sensitivity of conclusions to unmeasured confounding. The classical Cornfield condition allows us to assess whether an unmeasured binary confounder can explain away the observed relative risk of the exposure on the outcome. 
It states that for an unmeasured confounder to explain away an observed relative risk, the association between the unmeasured confounder and the exposure, and also that between the unmeasured confounder and the outcome, must both be larger than the observed relative risk. In this paper, we extend the classical Cornfield condition in three directions. 
First, we consider analogous conditions for the risk difference, and allow for a categorical, not just a binary, unmeasured confounder. 
Second, we provide more stringent thresholds which the maximum of the above-mentioned associations must satisfy, rather than simply weaker conditions that both must satisfy. 
Third, we show that all previous results on Cornfield conditions hold under weaker assumptions than previously used.
We illustrate their potential applications by real examples, where our new conditions give more information than the classical ones. 
\end{abstract}


\section{Causation, Confounding and Cornfield Question}

Causal inference in observational studies is often jeopardized by unmeasured confounding. 
For example, it can be the case that the crude association between the exposure and the outcome is positive, but their association is negative within each stratum of a confounder.  
This phenomenon is sometimes referred to as the Yule--Simpson paradox, and within the context of causal inference is referred to as confounding.
For instance, the association between cigarette smoking and lung cancer was historically accounted for by two competing theories: cigarette smoking has a causal effect on lung cancer; or cigarette smoking is not causative for lung cancer, and their crude association is purely due to a common genetic cause that influences both of them. 
R. A. Fisher was a major proponent of the second theory, viewing the Yule--Simpson paradox as an Achilles' heel of causal inference in observational studies \citep{Fisher::1957}.    
Instead of taking a completely dismissive view on observational studies, \citet{Cornfield::1959} asked the following question regarding the common cause theory:
How strong should the unmeasured confounder be, in order to explain away the association between the exposure and outcome?

If the strength of the unmeasured confounder required to explain away the association turns out to be too large to make biological or practical sense, then the association between the exposure and the outcome cannot be explained only by a common cause, and the exposure-response relationship itself must be causal. 
\citet{Cornfield::1959} settled the causal role of smoking on lung cancer using this approach, and their seminal work helped initiate the entire field of sensitivity analysis.
Here, we extend \citet{Cornfield::1959}'s work to the risk difference scale, to non-binary confounders, to stronger thresholds, and to weaker assumptions.

\section{Notation and Definitions}

Assume we have a binary exposure $E$, a binary outcome $D$, and a categorical unmeasured confounder $U$.
The discussion throughout the paper will make reference to the ignorable treatment assignment assumption $E \ind \{ D(1), D(0)\}\mid U$ \citep{Rosenbaum::1983}; where $D_i(1)$ and $D_i(0)$ denote the potential outcomes for individual $i$ with and without exposure, respectively, and $X\ind Y\mid Z$ is used to denote that $X$ is independent of $Y$ conditional on $Z$.
This is also sometimes referred to as the assumption that the effect of $E$ on $D$ is unconfounded conditional on $U.$
In order to answer the Cornfield question, we assume $E$ has no effect on $D$, and consider how large the associations between $U$ and $E$ and between $U$ and $D$ would have to be to explain away the observed crude association between $E$ and $D$.
There are three possible formulations about what might be meant by no effect of $E$ on $D$, and we will present them from weakest to strongest.

\begin{assumption}
The average causal effect of $E$ on $D$ is zero, i.e.,
$$
\sum_{k=0}^{K-1} \left\{    \pr(D=1\mid E=1, U=k) - \pr(D=1\mid E=0, U=k)  \right\} \pr(U=k) = 0.
$$ 
\end{assumption}

\begin{assumption}
The average causal effect of $E$ on $D$ is zero in every stratum of $U$, i.e.,
the exposure $E$ is conditionally independent of the outcome $D$ given $U.$
\end{assumption}

Assumption 2 implies that $ \pr(D=1\mid E=1, U=k) = \pr(D=1\mid E=0, U=k) $ and thus entails Assumption 1.

\begin{assumption}
The causal effect of $E$ on $D$ is zero for every individual in the population, i.e., $D_i(1)=D_i(0)$ for each individual $i$.
\end{assumption}

Under the ignorability assumption,
Assumption 3 implies that $\pr(D=1\mid E=1, U=k) =  \pr\{ D(1) = 1\mid E=1, U=k \} = \pr\{ D(1)=1\mid U=k \} = \pr\{  D(0)=1\mid U=k \} = \pr\{ D(0)\mid U=k, E=0 \} = \pr(D=1\mid E=0, U=k)$ and thus Assumption 2.

The previous literature on Cornfield conditions are derived under the conditional independence of the exposure $E$ and the outcome $D$ given the confounder $U$, i.e., $E\ind D\mid U$.
The ignorability assumption involving potential outcomes guarantees the causal interpretation of Assumptions 2 and 3.
Analogous assumptions can also be made using the causal diagram framework \citep{Pearl::1995}.
Without a formal causal framework, the Cornfield conditions under $E\ind D\mid U$ are the empirical conditions to explain away the crude exposure-outcome association by the association between $E$ and $U$ and that between $U$ and $D.$

\section{Cornfield Conditions for the Relative Risk with a Binary Confounder}

\citet{Cornfield::1959} derived their original conditions considering a binary confounder $U$.
Let $\RD_{ED} = \pr(D=1\mid E=1)-\pr(D=1\mid E=0)$ and $\RR_{ED} =\pr(D=1\mid E=1) / \pr(D=1\mid E=0)$ denote the risk difference and the relative risk of the exposure $E$ on the outcome $D$. The risk difference and relative risk for other variables can be defined analogously.
Without essential loss of generality, we assume $\RR_{ED}\geq 1$ and $\RR_{EU} \geq 1$. The conditions for preventive exposures are analogous. \citet{Cornfield::1959} showed that when
the confounder $U$ is binary and Assumption 2 holds, it must be true that
\begin{eqnarray}
\RR_{EU} \geq \RR_{ED}, \label{eq::corn-EU}
\end{eqnarray}
which means that the relative risk of $E$ on $U$ is greater than or equal to the relative risk of $E$ on $D$; and also \citep[cf.][]{Schlesselman::1978}
\begin{eqnarray}
\RR_{UD} \geq \RR_{ED},\label{eq::corn-UD}
\end{eqnarray} 
i.e., the relative risk of $U$ on $D$ is greater than or equal to the relative risk of $E$ on $D.$ 
Thus, for a binary unmeasured confounder to explain away an observed relative risk, the relative risk between the exposure and the unmeasured confounder and between the unmeasured confounder and the outcome must both be greater than the observed exposure-outcome relative risk. These necessary conditions under Assumption 2 are often now referred to as the classical Cornfield conditions.

\section{Generalized Cornfield Conditions for the Relative Risk}

\citet{Cornfield::1959} obtained their results for a binary confounder $U$ under Assumption 2. We show in the Supplementary Material that under the weaker Assumption 1 condition (\ref{eq::corn-EU}) still holds, and condition (\ref{eq::corn-UD}) can be replaced by
\begin{eqnarray*}
\max(\RR_{UD|E=1},\RR_{UD|E=0})\geq \RR_{ED},
\end{eqnarray*}
where $\RR_{UD|E=1}$ and $\RR_{UD|E=0}$ are the relative risk of $U$ on $D$ with and without exposure, respectively.

\citet{Lee::2011} discussed the case when $U$ is categorical with levels $0,1,\ldots, K-1$.
Define $p_k = \pr(E=1\mid U=k)$ as the probability of the exposure within $U=k$, $q_k = p_k/(1-p_k)$ as the odds of the exposure within $U=k$, and $U_E = \max_k q_k /\min_k q_k$ as the ratio of the maximum and the minimum of these odds.
Define
$ r_k = \pr(D=1\mid E=0, U=k)$ as the probability of the outcome without exposure and with $U=k$, and $  U_D = \max_k r_k/\min_k r_k $ as the ratio of the maximum and the minimum of these probabilities.
Similarly, define $r_k^* =  \pr(D=1\mid E=1, U=k)$ as the probability of the outcome with exposure and $U=k$, and $  U_D^* = \max_k r_k^*/\min_k r_k^* $. 
When $U$ is binary, $U_E$ reduces to the odds ratio between $E$ and $U$, $U_D$ reduces to $\RR_{UD|E=0}$, and $U_D^*$ reduces to $\RR_{UD|E=1}.$

Under Assumption 2, $r_k = r_k^*$, and so $U_D=U_D^*$.
Consequently, \citet{Lee::2011} showed that the Cornfield condition 
\begin{eqnarray}
\min(U_E, U_D)  &\geq &\RR_{ED} \label{eq::lee-min}
\end{eqnarray}
holds for categorical $U$. He also derived a generalized Cornfield condition for the maximum of $U_E$ and $U_D$:
\begin{eqnarray}
\max(U_E, U_D) &\geq &\left\{ \RR_{ED} ^{1/2} + ( \RR_{ED} - 1 )^{1/2}  \right\}^2, \label{eq::lee-max}
\end{eqnarray}
which gives a higher threshold than the Cornfield condition (\ref{eq::lee-min}) and thus can be more informative.
For example, an observed crude $\RD_{ED} = 1.5$ requires that both $U_E$ and $U_D$ be greater than $1.5$; these are essentially the classical Cornfield conditions. However, \cite{Lee::2011}'s generalized Cornfield conditions also require
that the maximum of $U_E$ and $U_D$ be at least as large as $(1.5^{1/2}+0.5^{1/2})^2=3.73$. Similarly, an observed crude $\RD_{ED}=5$ would require that the maximum of $U_E$ and $U_D$ be at least as large as $(5^{1/2}+4^{1/2})^2 = 17.94$.

\citet{Lee::2011} obtained the above results (\ref{eq::lee-min}) and (\ref{eq::lee-max}) under Assumption 3, which can in fact be weakened to Assumption 2.
Furthermore in the Supplementary Material, we show that under Assumption 1, the following conditions must hold:
\begin{eqnarray*}
\min(U_E, U_D') \geq \RR_{ED},\quad 
\max(U_E, U_D') \geq \left\{   \RR_{ED} ^{1/2} + (\RR_{ED} - 1 )^{1/2} \right\}^2,
\end{eqnarray*}
where $U_D' = \max(U_D, U_D^*)$ replaces $U_D$ in conditions (\ref{eq::lee-min}) and (\ref{eq::lee-max}).

\section{Generalized Cornfield Conditions for the Risk Difference}

Because of \citet{Cornfield::1959}'s influential work, sensitivity analysis based on ratio measures has long been
predominant in epidemiology. 
However, sometimes the risk difference is of interest and sensitivity analysis can likewise be conducted using the risk difference.
\citet{Poole::2010} moreover showed, via a real example, that dismissal of the risk difference in sensitivity analysis will sometimes restrict evidence for important scientific findings, a point that we will discuss later.

For the risk difference, we define $\alpha_k = \pr(U=k\mid E=1) - \pr(U=k\mid E=0)$, and $A = \max_{k\geq 1} |\alpha_k|$.
The parameter $\alpha_k$ measures the difference in the probability that $U$ takes a particular value $k$ comparing exposed and unexposed, and $A$ is the maximum of these absolute differences.
Define $\beta_{1k} = \pr(D=1\mid E=1, U=k) - \pr(D=1\mid E=1, U=0), \beta_{0k} = \pr(D=1\mid E=0, U=k) - \pr(D=1\mid E=0, U=0)$, and $B = \max( \max_{k\geq 1}  | \beta_{1k} | ,  \max_{k\geq 1} |  \beta_{0k} | )$.
The parameters $\beta_{1k}$ and $\beta_{0k}$ measure the difference in the probability of the outcome comparing category of $U=k$ to $U=0$ in the exposed and unexposed groups, respectively; and $B$ is the maximum of these absolute differences.
When Assumption 2 holds, $\beta_{1k}  = \beta_{0k} = \beta_k = \pr(D=1\mid U=k)$ and $B=\max_{k\geq 1}  | \beta_k| .$

For a binary confounder $U$ with categories $0$ and $1$, $A = \RD_{EU} $ is the risk difference of $E$ on $U$, and
$B = \max(|\RD_{UD|E=1}|, |\RD_{UD|E=0}|)$ is the maximum of the absolute values of the risk differences of $U$ on $D$ with and without exposure.
Under Assumption 2, since $\RD_{UD|E=1} = \RD_{UD|E=0}$, $B = \RD_{UD} $ is the risk difference of $U$ on $D$. The generalized Cornfield conditions for the risk difference with a binary confounder are shown below.

\begin{theorem}
If the confounder $U$ is binary with $K=2$ and Assumption 1 holds, then
\begin{eqnarray}
\min\left\{    \RD_{EU} ,    \max(  \RD_{UD|E=1},  \RD_{UD|E=0}  ) \right\} &\geq& \RD_{ED},   \label{eq::corn-01-1-neyman} \\
\max\left\{   \RD_{EU} ,    \max(  \RD_{UD|E=1},  \RD_{UD|E=0}  ) \right\} &\geq&  \RD_{ED}^{1/2}. \label{eq::corn-01-2-neyman}
\end{eqnarray}
\end{theorem}

Under Assumption 2, conditions (\ref{eq::corn-01-1-neyman}) and (\ref{eq::corn-01-2-neyman}) can be further simplified to
\begin{eqnarray}
\min(   \RD_{EU}  ,  \RD_{UD}   ) &\geq& \RD_{ED}, \label{eq::corn-01-1}\\
\max(   \RD_{EU} ,  \RD_{UD}  )  &\geq& \RD_{ED}^{1/2} . \label{eq::corn-01-2}
\end{eqnarray}

The Cornfield conditions for the risk difference for a binary confounder thus show that for an unmeasured confounder to explain away an observed risk difference for $E$ on $D$, the risk difference for $E$ on $U$ and also that for $U$ on $D$ must both be larger than the observed risk difference for $E$ on $D$. Moreover, the maximum of these two risk differences for the unmeasured confounder must be greater than the square root of the observed risk difference for $E$ on $D$.

\citet{Cornfield::1959} obtained, but did not appreciate the significance of equation $\RD_{ED} =\RD_{EU}  \RD_{UD}$, which leads to conditions (\ref{eq::corn-01-1}) and (\ref{eq::corn-01-2}).
\citet{Gastwirth::1998} and \citet{Poole::2010} discussed the first Cornfield condition (\ref{eq::corn-01-1}) for the risk difference in the presence of a binary confounder, and the second one (\ref{eq::corn-01-2}) is new to the best of our knowledge. 
Although quite simple,
the new square root bounds (\ref{eq::corn-01-2-neyman}) and (\ref{eq::corn-01-2}) can be substantial improvements over (\ref{eq::corn-01-1-neyman}) and (\ref{eq::corn-01-1}),
since $\RD_{ED}$ is very small in many applications.

We can further extend these Cornfield conditions for the risk difference to allow for a categorical, rather than binary, confounder $U$ with an arbitrary number of categories.
However, the generalized Cornfield conditions for the risk difference then depend on the number of categories of the confounder $U$. With more than two categories, we have the following conditions.

\begin{theorem}\label{thm::rd-general}
With a categorical confounder $U (K\geq 3)$, if Assumption 1 holds,
\begin{eqnarray}
A &\geq& \RD_{ED}/(K-1),\label{eq::bound-A}\\
B&\geq& \RD_{ED}/2 , \label{eq::bound-B}\\
\max(A, B) &\geq& \max\left[   \{ \RD_{ED}/(K-1) \}^{1/2},  \RD_{ED}/2 \right]  .\label{eq::bound-AB}
\end{eqnarray}
\end{theorem}

Theorem \ref{thm::rd-general} does not rely on the choice of the reference level of $U$, so continues to hold if we arbitrarily relabel some other level
to be $U=0.$
Inequalities (\ref{eq::bound-A}) to (\ref{eq::bound-AB}) show that the conditions for $A$ and $B$ become weaker with a larger value of $K.$
For example, if $U$ has three categories with $K=3$, as would often be the case with a genetic confounder, the generalized Cornfield conditions above can be simplified as
\begin{eqnarray}\label{eq::k=3}
\min(A, B)  \geq  \RD_{ED}/2,
\quad 
\max(A,B) \geq  ( \RD_{ED}/2 )^{1/2}.
\end{eqnarray}
Although the conditions above are weaker when $U$ has three categories, the lower bound of the maximum of $A$ and $B$, $(\RD_{ED}/2 )^{1/2}$, can still be very informative even if $\RD_{ED}$ is small.

In many practical problems, the following monotonicity assumption is plausible.

\begin{assumption}
\label{assume::mono}
For $k=1,\ldots,(K-1)$, $\alpha_k\geq 0$.
\end{assumption}

Assumption 4 requires that each non-zero category of $U$ is more prevalent under exposure than without the exposure.
If only one category of $U$ is less prevalent under exposure, Assumption \ref{assume::mono} holds if we choose this category to be the reference level $U=0.$
For example, Assumption \ref{assume::mono} holds for a binary confounder $U$ without imposing any restrictions.

\begin{theorem}\label{thm::rd::general-mono}
With a categorical confounder $U (K\geq 2)$, under Assumptions 1 and 4,
\begin{eqnarray}
A &\geq& \RD_{ED}/(K-1),\label{eq::bound-A-mono}\\
B&\geq& \RD_{ED}, \label{eq::bound-B-mono}\\
\max(A, B) &\geq& \max\left[  \{ \RD_{ED}/(K-1) \}^{1/2},  \RD_{ED}\right]  .\label{eq::bound-AB-mono}
\end{eqnarray}
\end{theorem}

For example, the conditions in (\ref{eq::k=3}) with a three category confounder can be improved to
\begin{eqnarray}
\label{eq::k=3m}
A\geq \RD_{ED}/2, B\geq \RD_{ED},  \quad \max(A,B)\geq \max\left\{  (\RD_{ED}/2)^{1/2}, \RD_{ED} \right\} .
\end{eqnarray}

We call (\ref{eq::corn-01-1-neyman}) to (\ref{eq::k=3m}) the generalized Cornfield conditions for the risk difference.
The bounds from (\ref{eq::corn-01-1-neyman}) to (\ref{eq::k=3m}) are sharp, in the sense that they cannot be improved without additional assumptions. 
The bounds given above for Assumption 2 are all also sharp under Assumption 3.
The proofs for attaining the bounds are all given in the Supplementary Material.

\section{Illustrations}

\begin{example}
\citet{Poole::2010} discussed an example concerning insights from the Confield conditions for the risk difference, which were overlooked by \citet{Cornfield::1959}.
\citet{Hammond::1958}'s study of smoking and death from coronary artery disease found that $\RR_{ED}=1.7$ and $\RD_{ED}= 0.013\%$. 
Based on the relative risks, there is limited evidence for a true causal of smoking on death, since the crude relative risk of smoking on death $\RR_{ED}=1.7$
is smaller than the relative risk of bad temper in smokers and nonsmokers $\RR_{EU}=2.6$ found in \citet{Lilienfeld::1959}'s study.
However, the risk difference of the exposure on bad temper is smaller than the observed risk difference of the exposure on the outcome, since $\RD_{EU}= 0.012\%<\RD_{ED}$.
\citet{Poole::2010} noted that
if we assume $U=$ bad temper, then the common cause $U$ could not explain away the risk difference of smoking on death from coronary artery diseases.

Poole's observation was very insightful. Unfortunately, however, the difference between the exposure-confounder association $\RD_{EU} =0.012\%$ and the actual observed exposure-outcome association $\RD_{ED} =0.013\%$ is very small and likely within sampling variabilities of these studies. In fact, \citet{Poole::2010} noted that $\RD_{ED}=0.012\%$ in \citet{Doll::1964}'s study, and then the basic Cornfield condition $\RD_{EU}\geq \RD_{ED}$ would not be sufficient
to reject the common cause theory. However, if we were willing to assume that 
the risk difference between smoking and bad temper is stronger than the risk difference between
bad temper and death, i.e., $\RD_{EU} > \RD_{UD|E=1}$ and $\RD_{EU} > \RD_{UD|E=0} $, then by the generalized Cornfield conditions for the measures of $\RD_{EU}$ and $\max(\RD_{UD|E=1}, \RD_{UD|E=0})$ given in (\ref{eq::corn-01-2-neyman}), for the unmeasured common cause bad temper to explain away the effect, 
 we would require that $\RD_{EU} = \max( \RD_{EU}, \RD_{UD|E=1}, \RD_{UD|E=0}  ) \geq \RD_{ED}^{1/2} =  (0.012\%)^{1/2} = 1.095\%$, which is more than $84$ times larger than the observed $\RD_{EU}$.
The confounder bad temper could then not explain away the association, and the evidence for causation would be very strong.
\end{example}

\begin{example}
R. A. Fisher conjectured that ``cigarette-smoking and lung cancer, though not mutually causative, are both influenced by a common cause, in this case the individual genotype'' \citep{Fisher::1957}. 
Consider the case that the genotype has three categories (AA, Aa, aa), where the Cornfield condition with a binary confounder does not apply.
From \citet{Hammond::1958}'s study, the relative risk and the risk difference of smoking on lung cancer are $\RR_{ED} = 10.7$ and $\RD_{ED} = 0.094\%.$
If we want to assume no average causal effect of smoking on lung cancer,
the generalized Cornfield conditions for the risk ratio require that $\min(U_E, U_D') \geq \RR_{ED} = 10.7,$ and $ \max(U_E, U_D') \geq \{ \RR_{ED}^{1/2} + (\RR_{ED}-1)^{1/2}\}^2 =  ( 10.7^{1/2} + 9.7^{1/2})^2 =  40.77$. Without Assumption 4, our conditions for the risk difference require $ A\geq \RD_{ED}/2 =  0.047\%, B\geq \RD_{ED}/2 =  0.047\%,$ and $ \max(A, B)\geq (\RD_{ED}/2)^{1/2} =  2.168\%$, and with Assumption 4, the condition for $B$ can be further improved to $B\geq \RD_{ED} =0.094\%.$
These bounds are all useful for determining whether a certain genotype can explain away the association between smoking and lung cancer.
\end{example}

\section{Discussion}

The risk difference scale can sometimes be used in sensitivity analysis.
More specifically, as
pointed out by \citet{Poole::2010}, the Cornfield conditions for the risk difference can be useful for discovering causal effects in observational studies. Our new lower bound of the maximum of $\RD_{EU}$ and $\max(\RD_{UD|E=1}, \RD_{UD|E=0})$ is a significant improvement of the basic Cornfield condition for the risk difference with a binary confounder, when the risk difference of the exposure on the outcome is small.
We also illustrate its usefulness in Example 1, where the lower bound may provide a sharper conclusion.

The results in our paper are also interesting from a theoretical perspective in two further respects. First, it has been argued 
that causal conclusions are most sensitive to an unobserved confounder that is binary rather than categorical or continuous, at least when using matched pairs analysis. It has been suggested that
it thus suffices to perform sensitivity analysis with a binary confounder \citep{Wang::2006}.
However, these results on the conservative nature of a binary confounder were derived with sensitivity analysis parameters expressed on a ratio scale. In this paper, we have likewise seen that using a ratio scale the Cornfield conditions for a categorical unmeasured confounder in (3) are essentially identical to those with a binary unmeasured confounder in (1) and (2). For relative risks, a binary unmeasured confounder seems to once again suffice. However, our 
results here for the risk difference demonstrate that the generalized Cornfield conditions for the risk difference do depend on the number of categories of the unmeasured confounder $U$. 
The requirements on the unmeasured confounder weaken as the number of categories of $U$ increases. Thus the sensitivity of the causal conclusions when the unmeasured confounder is binary is not the most conservative case if the sensitivity analysis parameters are expressed on the risk difference scale.

Second, since \citet{Cornfield::1959}'s seminal work, the relative risk measure
has often been claimed to be better suited for assessing causality. \citet{Poole::2010} recently summarized the historical reasons for this and also criticized this notion. 
Our results demonstrate that the generalized Cornfield conditions for the risk difference do depend on the number of categories of $U$, while those for the relative risk do not.
The Cornfield conditions for the risk difference become less informative as the number of categories increases.
The Cornfield conditions for the relative risk do not suffer from this problem.
Therefore, using the relative risk for assessing causality may in fact have some basis, because the generalized Cornfield conditions for the relative risk do not depend on the number of categories of $U$.

%
%
%
%


\begin{thebibliography}{7}
\expandafter\ifx\csname natexlab\endcsname\relax\def\natexlab#1{#1}\fi


\bibitem[{Cornfield et al.(1959)}]{Cornfield::1959}
\textsc{Cornfield J., Haenszel W., Hammond E.C., et al. } (1959).
\newblock {Smoking and lung cancer: recent evidence and a discussion of some questions}.
\newblock \textit{Journal of the National Cancer Institute} \textbf{22}, 173--203.




\bibitem[{Doll and Hill(1964)}]{Doll::1964}
\textsc{Doll, R. \& Hill, A. B.} (1964).
\newblock {Mortality in relation to smoking: ten years' observations of British doctors.}
\newblock \textit{British Medical Journal} \textbf{1}, 1399--1410.






\bibitem[{Fisher(1957)}]{Fisher::1957}
\textsc{Fisher R. A.} (1957).
\newblock {Dangers of cigarette smoking [letter]}.
\newblock \textit{British Medical Journal} \textbf{2}, 297--298.



\bibitem[{Gastwirth et al. (1998)}]{Gastwirth::1998}
\textsc{Gastwirth, J. L., Krieger, A. M., and Rosenbaum, P. R.} (1998).
\newblock {Cornfield's inequality}.
\newblock \textit{In Encyclopedia of Biostatistics} 952--955. Wiley, New York.





\bibitem[{Hammond and Horn(1958)}]{Hammond::1958}
\textsc{Hammond, E. C. \& Horn, D.} (1958).
\newblock {Smoking and death rates: report on forty four months of follow-up of $187,783$ men. } 
\newblock \textit{Journal of the American Medical Association} \textbf{166}, 1159--1172,	1294--1308.





\bibitem[{Lee(2011)}]{Lee::2011}
\textsc{Lee, W. C.} (2011).
\newblock {Bounding the bias of unmeasured factors with confounding and effect-modifying potentials}.
\newblock \textit{Statistics in Medicine} \textbf{30}, 1007--1017.





\bibitem[{Lilienfeld(1959)}]{Lilienfeld::1959}
\textsc{Lilienfeld, A. M.} (1959).
\newblock {Emotional and other selected characteristics of cigarette smokers and nonsmokers as related to epidemiological studies of lung cancer and other diseases. }
\newblock \textit{Journal of the National Cancer Institute} \textbf{22}, 259--282.



\bibitem[{Pearl(1995)}]{Pearl::1995}
\textsc{Pearl, J.} (1995). 
\newblock {Causal diagrams for empirical research (with discussion).} 
\newblock \textit{Biometrika} {\bfseries 82}, 669--688.




\bibitem[{Poole(2010)}]{Poole::2010}
\textsc{Poole, C.} (2010).
\newblock {On the origin of risk relativism.}
\newblock \textit{Epidemiology} \textbf{21}, 3--9.




\bibitem[{Rosenbaum and Rubin(1983)}]{Rosenbaum::1983}
\textsc{Rosenbaum, P. and Rubin, D. B.} (1983).
\newblock {The central role of the propensity score in observational studies for causal effects.}
\newblock \textit{Biometrika} \textbf{70}, 41--55.



\bibitem[{Schlesselman(1978)}]{Schlesselman::1978}
\textsc{Schelesselman, J. J.} (1978).
\newblock {Assessing effects of confounding variables.}
\newblock \textit{American Journal of Epidemiology} \textbf{108}, 3--8.







\bibitem[{Wang and Krieger(2006)}]{Wang::2006}
\textsc{Wang, L. \& Krieger, A. M.}(2006).
\newblock {Causal conclusions are most sensitive to unobserved binary covariates.}
\newblock \textit{Statistics in Medicine} \textbf{25}, 2257--2271.



\end{thebibliography}

\newpage

\begin{center}
\bfseries \Large Supplementary Materials
\end{center}

\section*{Appendix A}

This Appendix gives a proof of the Cornfield conditions for the relative risk with a binary confounder under Assumption 1.
\begin{proof}
Define $f = \pr(U=1), p_e = \pr(E=1) , f_1 = \pr(U=1\mid E=1),$ and $f_0  = \pr(U=1\mid E=0)$.
We have $f = p_e  f_1 + (1 -  p_e ) f_0 $, and we assume $\RR_{EU} = f_1/f_0 \geq 1$.
Recall the definitions of $r_k^* = \pr(D=1\mid E=1, U=k)$ and $r_k = \pr(D=1\mid E=0, U=k)$ in the main text.
For simplicity in the proof, we use $\RR_1  = \RR_{UD|E=1} = r_1^* / r_0^* $ and $\RR_0 = \RR_{UD|E=0} = r_1 / r_0$ as the relative risks of $U$ on $D$ given $E=1$ and $E=0$, respectively.

Assumption 1 
$$
1 =
\frac{    \pr(D=1\mid E=1, U=1) \pr(U=1) + \pr(D=1\mid E=1, U=0) \pr(U=0)    }{     \pr(D=1\mid E=0, U=1) \pr(U=1) + \pr(D=1\mid E=0, U=0) \pr(U=0)   } 
$$
is equivalent to
\begin{eqnarray}\label{eq::neyman-null}
1 = 
\frac{   r_1^*  f + r_0^* (1-f)  }{  r_1  f + r_0  (1-f)    }
=  \frac{r_0^*}{r_0}  \times   \frac{ \RR_1 f + (1-f)  }{ \RR_0 f + (1-f) }.
\end{eqnarray}
Therefore, the observed relative risk of $E$ on $D$ 
$$
\RR_{ED} =
\frac{    \pr(D=1\mid E=1, U=1) \pr(U=1\mid E=1) + \pr(D=1\mid E=1, U=0) \pr(U=0\mid E=1)    }{     \pr(D=1\mid E=0, U=1) 
\pr(U=1\mid E=0) + \pr(D=1\mid E=0, U=0) \pr(U=0\mid E=0)   } 
$$
can be expressed as
\begin{eqnarray*}
\RR_{ED}= \frac{  r_1^* f_1 + r_0^* (1-f_1)  }{   r_1 f_0 + r_0 (1-f_0)  } 
= \frac{r_0^*}{r_0}  \times \frac{   \RR_1 f_1 + (1-f_1)  }{ \RR_0 f_0 + (1-f_0) } 
=  \frac{ \RR_0 f + (1-f) } { \RR_1 f + (1-f)  }   \times \frac{   \RR_1 f_1 + (1-f_1)  }{ \RR_0 f_0 + (1-f_0) } .
\end{eqnarray*}
The last equation above is obtained by replacing $ r_0^* / r_0 $ by $\{ \RR_0 f + (1-f) \} /  \{ \RR_1 f + (1-f)  \}  $ due to (\ref{eq::neyman-null}).
The above equation can be further simplified as 
$$
\RR_{ED} = G \times H,$$ 
where
\begin{eqnarray*}
G =   \frac{ \RR_0 f + (1-f) } { \RR_1 f + (1-f)  } 
\quad \text{and} \quad 
H = \frac{   \RR_1 f_1 + (1-f_1)  }{ \RR_0 f_0 + (1-f_0) } .
\end{eqnarray*}

We first treat $(\RR_1, \RR_0 , f_1, f_0)$ as fixed, and thus $G$ is a function of $p_e$ with partial derivative
\begin{eqnarray*}
\frac{\partial G }{ \partial  p_e}  
 =  \frac{\partial }{\partial f}  \left\{  \frac{ \RR_0 f + (1-f) } { \RR_1 f + (1-f)  }  \right\} 
\times   \frac{\partial f}{\partial p_e}  
=  \frac{  (\RR_0 - \RR_1)( f_1 - f_0)   }{    \{ \RR_1 f + (1-f)  \}^2    } .
\end{eqnarray*}
Therefore, $G$ is increasing in $p_e\in [0,1]$ if $\RR_0 > \RR_1$, and non-increasing in $p_e\in [0,1]$ if $\RR_0 \leq \RR_1$. Our proof below is divided into two cases accordingly.

If $\RR_0 > \RR_1$, $G$ has its maximum at $p_e = 1$ or $f=f_1.$ Therefore,
\begin{eqnarray}
\label{eq::rr0-larger}
\RR_{ED} \leq  \frac{ \RR_0 f_1 + (1-f_1) } { \RR_1 f_1 + (1-f_1)  } \times \frac{   \RR_1 f_1 + (1-f_1)  }{ \RR_0 f_0 + (1-f_0) }
= \frac{  (\RR_0 - 1) f_1 +1    }{ (\RR_0  - 1) f_0  + 1  } .
\end{eqnarray}
With $f_1 \geq  f_0$, we must have $\RR_0 \geq 1$, since $\RR_0 < 1$ would contradict the assumption $\RR_{ED} \geq 1$. 
Then (\ref{eq::rr0-larger}) attains its maximum at $f_1=1$ and $f_0=0$, implying that $\RR_{ED} \leq \RR_0 = \max(\RR_1, \RR_0)$.
We further obtain from (\ref{eq::rr0-larger}) that
\begin{eqnarray}
\label{eq::rr0-larger-2}
\RR_{ED} \leq  \frac{f_1}{f_0} \times \frac{1 + (1-f_1)/\RR_0 }{1 + (1-f_0)/\RR_0}  \leq \frac{f_1}{f_0}  = \RR_{EU},
\end{eqnarray}
where the second inequality in (\ref{eq::rr0-larger-2}) holds since $f_1 \geq  f_0$.

If $\RR_0 \leq \RR_1$, $G$ has its maximum at $p_e = 0$ or $f = f_0.$ Therefore, 
\begin{eqnarray}
\label{eq::rr0-smaller}
\RR_{ED} \leq \frac{ \RR_0 f_0 + (1-f_0) } { \RR_1 f_0 + (1-f_0)  } \times \frac{   \RR_1 f_1 + (1-f_1)  }{ \RR_0 f_0 + (1-f_0) }
= \frac{  (\RR_1 - 1) f_1 +1    }{ (\RR_1  - 1) f_0  + 1  } .
\end{eqnarray}
By similar argument, we must have $\RR_1 \geq 1$, and the right-hand side of (\ref{eq::rr0-smaller}) attains its maximum at $f_1=1$ and $f_0=0$. 
Therefore $\RR_{ED} \leq \RR_1 = \max(\RR_1, \RR_0)$. The same argument as above shows that $\RR_{ED} \leq \RR_{EU}$.

In summary, we have shown that $\max(\RR_1, \RR_0)  \geq  \RR_{ED} $ and $\RR_{EU} \geq \RR_{ED}$ in all cases.
\end{proof} 

\section*{Appendix B}

This Appendix gives a proof of \cite{Lee::2011}'s conditions for relative risk under Assumption 1.
Our proof here is based on \cite{Lee::2011}'s notation and conclusions.
Define 
$$
\SRR^{E+} = \frac{\sum_{k=0}^{K-1}  \pr(U=k)\pr(E=1\mid U=k) \pr(D=1\mid E=1, U=k)  }{ \sum_{k=0}^{K-1}  \pr(U=k)\pr(E=1\mid U=k) \pr(D=1\mid E=0, U=k)  }
$$
as the standardized relative risk with the exposed group taken as the standard population,
$$
\SRR^{E-} = \frac{\sum_{k=0}^{K-1}  \pr(U=k)\pr(E=0 \mid U=k) \pr(D=1\mid E=1, U=k)  }{ \sum_{k=0}^{K-1}  \pr(U=k)\pr(E=0\mid U=k) \pr(D=1\mid E=0, U=k)  }
$$
as the standardized relative risk with the unexposed group taken as the standard population, and
$$
\SRR^T =  \frac{\sum_{k=0}^{K-1}  \pr(U=k)  \pr(D=1\mid E=1, U=k)  }{ \sum_{k=0}^{K-1}  \pr(U=k)  \pr(D=1\mid E=0, U=k)  }
$$
as the standardized relative risk with the total group as the standard population.
And the confounding relative risks are defined as $\CRR^{E+} = \RR_{ED}/\SRR^{E+},\CRR^{E-} = \RR_{ED}/\SRR^{E-}$, and $\CRR^T = \RR_{ED}/\SRR^T.$
Lee (2011) showed that 
$$
\frac{1}{\CRR^T} = \frac{w}{ \CRR^{E+}  } + \frac{1-w}{ \CRR^{E-} }, 
$$
where $w$ is a positive number between $0$ and $1$.
The following conclusions in \cite{Lee::2011} are useful for our proof:
\begin{eqnarray*}
\CRR^{E+} \leq  \left\{   \frac{  (U_E U_D)^{1/2} + 1}{ U_E^{1/2}  + U_D^{1/2} } \right\}^2,
\quad 
\CRR^{E-} \leq   \left\{  \frac{  (U_E U_D^*)^{1/2} + 1}{ U_E^{1/2}  + U_D^{*1/2} } \right\}^2 ,
\end{eqnarray*}

\begin{proof}
It can be directly verified that 
$$
 \frac{  (U_E U_D)^{1/2} + 1}{ U_E^{1/2}  + U_D^{1/2} }$$
is increasing in both $U_D$ and $U_E$. For example, we have 
\begin{eqnarray*}
\frac{\partial }{\partial U_D} \left\{    \frac{  (U_E U_D)^{1/2} + 1}{ U_E^{1/2}  + U_D^{1/2} }  \right\}  
&=& 
\frac{    U_E^{1/2}  U_D^{-1/2}  (U_E^{1/2}  + U_D^{1/2})   -
\{  (U_E U_D)^{1/2} + 1 \} U_D^{-1/2}  }
{  2  ( U_E^{1/2}  + U_D^{1/2})^2   }\\
&=&
\frac{   U_E - 1   }{  2U_D^{1/2}   (U_E^{1/2}  + U_D^{1/2})^2     } \geq 0.
\end{eqnarray*}

By definition of $U_D'$ and 
according to \cite{Lee::2011}, we have
$$
\CRR^{E+}  \leq  \left\{  \frac{  (U_E U_D')^{1/2} + 1}{ U_E^{1/2}  + U_D^{'1/2} } \right\}^2
\text{ and }
\CRR^{E-} \leq  \left\{  \frac{  (U_E U_D')^{1/2} + 1}{ U_E^{1/2}  + U_D^{'1/2} } \right\}^2,
$$
which lead to
$$
\frac{1}{\CRR^T}  = \frac{\SRR^T}{\RR_{ED}} = \frac{w}{\CRR^{E+}} + \frac{1-w}{\CRR^{E-}} \geq   
\left\{  \frac{  (U_E U_D')^{1/2} + 1}{ U_E^{1/2}  + U_D^{'1/2} } \right\}^{-2}.
$$
When Assumption 1 holds with $\SRR^T = 1$, we have
$$
 \left\{  \frac{  (U_E U_D')^{1/2} + 1}{ U_E^{1/2}  + U_D^{'1/2} } \right\}^2 \geq   \RR_{ED}. 
$$
Letting $U_D'\rightarrow +\infty$ on the left-hand side of the last equation, we have
$U_E \geq \RR_{ED}$. By symmetry, we have $U_D' \geq \RR_{ED}$. And therefore, 
$
\min(U_E, U_D') \geq \RR_{ED}.
$
By monotonicity, we have
$$
 \left\{ \frac{  \max( U_E, U_D' )  + 1}{ 2 \max^{1/2}( U_E, U_D' )  } \right\}^2
  \geq    \left\{  \frac{  (U_E U_D')^{1/2} + 1}{ U_E^{1/2}  + U_D^{'1/2} } \right\}^2
   \geq   \RR_{ED},
$$
which implies that
$
\max( U_E, U_D' )  \geq \left\{  \RR_{ED}^{1/2}  + (\RR_{ED} - 1)^{1/2}  \right\}^2. 
$
Therefore, \cite{Lee::2011}'s conditions hold for $U_E$ and $U_D'.$
\end{proof}

\section*{Appendix C}

This Appendix gives proofs of the generalized Cornfield conditions for the risk difference under Assumption 1.
In order to prove Theorems 1 to 3, we need the following lemma.

\begin{lemma}
Under Assumption 1, the risk difference of $E$ on $D$ can be expressed as
$$
\RD_{ED} =  \sum_{k=1}^{K-1} \alpha_k  \{  \beta_{1k} \pr(E=0) + \beta_{0k} \pr(E=1) \}.
$$
\end{lemma}

\begin{proof}[of Lemma 1]
First, Assumption 1 is equivalent to 
$$
\sum_{k=0}^{K-1}     \pr(D=1\mid E=1, U=k) \pr(U=k)   = \sum_{k=0}^{K-1}   \pr(D=1\mid E=0, U=k)  \pr(U=k)  ,
$$
and therefore we have
\begin{eqnarray*}
\RD_{ED} 
&=& \sum_{k=0}^{K-1} \pr( D=1\mid E=1, U=k ) \pr(U=k\mid E=1)  \\
&&- \sum_{k=0}^{K-1} \pr(D=1\mid E=0, U=k) \pr(U=k\mid E=0)\\
&=&  \sum_{k=0}^{K-1} \pr( D=1\mid E=1, U=k ) \{   \pr(U=k\mid E=1) - \pr(U=k) \} \\
&& -\sum_{k=0}^{K-1} \pr(D=1\mid E=0, U=k) \{  \pr(U=k\mid E=0)  - \pr(U=k) \}  .
\end{eqnarray*}
Applying the law of total probability, we have that 
\begin{eqnarray*}
&&\pr(U=k\mid E=1) - \pr(U=k) \\
&=& \pr(U=k\mid E=1) - \pr(U=k\mid E=1)\pr(E=1) - \pr(U=k\mid E=0) \pr(E=0)  \\
&=& \{\pr(U=k\mid E=1) - \pr(U=k\mid E=0)  \} \pr(E=0) \\
&=& \alpha_k  \pr(E=0) ,
\end{eqnarray*}
and similarly, $\pr(U=k\mid E=0) - \pr(U=k) =  - \alpha_k  \pr(E=1) . $
Therefore, 
\begin{eqnarray*}
\RD_{ED} &=& \sum_{k=0}^{K-1} \alpha_k \pr(D=1\mid E=1, U=k) \pr(E=0) +  \sum_{k=0}^{K-1} \alpha_k \pr(D=1\mid E=0, U=k) \pr(E=1)\\
&=&\sum_{k=0}^{K-1} \alpha_k  \{    \pr(D=1\mid E=1, U=k) \pr(E=0) +    \pr(D=1\mid E=0, U=k) \pr(E=1)   \}.
\end{eqnarray*}
Using the fact that $\alpha_0 = -\sum_{k=1}^{K-1} \alpha_k$, we obtain that
\begin{eqnarray*}
\RD_{ED} &=&\sum_{k=1}^{K-1} \alpha_k  \{    \pr(D=1\mid E=1, U=k) \pr(E=0) +    \pr(D=1\mid E=0, U=k) \pr(E=1)   \} \\
&&- \sum_{k=1}^{K-1} \alpha_k  \{    \pr(D=1\mid E=1, U=0) \pr(E=0) +    \pr(D=1\mid E=0, U=0) \pr(E=1)   \} \\
&=& \sum_{k=1}^{K-1} \alpha_k  \{  \beta_{1k} \pr(E=0) + \beta_{0k} \pr(E=1) \}.
\end{eqnarray*}
\end{proof}

\begin{proof}[of Theorem 1]
For a binary confounder $U$ with $K=2$, we have
\begin{eqnarray*}
\RD_{ED} &=& \alpha_1 \{ \beta_{11} \pr(E=0) + \beta_{01} \pr(E=1) \} \\
&=& \RD_{EU} \{   \RD_{UD|E=1} \pr(E=0) + \RD_{UD|E=0}\pr(E=1)    \} .
\end{eqnarray*}

Since $\RD_{ED} \geq 0$ and $\RD_{EU} \geq 0$, we have $ \RD_{UD|E=1} \pr(E=0) + \RD_{UD|E=0}\pr(E=1)  \geq 0$. Evidently,
it is impossible that both $ \RD_{UD|E=1} $ and $ \RD_{UD|E=0}$ are negative. When $\RD_{UD|E=1}>0$ and $ \RD_{UD|E=0}>0$, we have 
$\RD_{UD|E=1} \pr(E=0) + \RD_{UD|E=0}\pr(E=1) < \max(  \RD_{UD|E=1}  ,  \RD_{UD|E=0} )$. When $\RD_{UD|E=1}>0$ and $\RD_{UD|E=0}<0$, we have $\RD_{UD|E=1} \pr(E=0) + \RD_{UD|E=0}\pr(E=1) <\RD_{UD|E=1}  =  \max(  \RD_{UD|E=1}  ,  \RD_{UD|E=0} )$. When 
$\RD_{UD|E=1} < 0$ and $ \RD_{UD|E=0} > 0$, we also have $\RD_{UD|E=1} \pr(E=0) + \RD_{UD|E=0}\pr(E=1)  <   \max(  \RD_{UD|E=1}  ,  \RD_{UD|E=0} )$. Therefore,
$$
 \RD_{ED}  \leq  \RD_{EU}  \times  \max( \RD_{UD|E=1} ,  \RD_{UD|E=0} ) , 
$$
which implies that
\begin{eqnarray*}
\min\left\{    \RD_{EU} ,    \max(  \RD_{UD|E=1}  ,  \RD_{UD|E=0}  )  \right\}  &\geq& \RD_{ED},\\
\max \left\{    \RD_{EU} ,   \max(  \RD_{UD|E=1} ,  \RD_{UD|E=0}  ) \right\}  &\geq&  \RD_{ED}^{1/2}.
\end{eqnarray*}
\end{proof}

\begin{proof}[of Theorem 2]
Since 
\begin{eqnarray*}
\RD_{ED} &=& \Big|  \sum_{k=1}^{K-1} \alpha_k \{  \beta_{1k} \pr(E=0) + \beta_{0k}\pr(E=1)  \} \Big|  \\
&\leq& A   \sum_{k=1}^{K-1}  |  \beta_{1k} \pr(E=0) + \beta_{0k}\pr(E=1)  |  \\
&\leq &   A\sum_{k=1}^{K-1}  \max( | \beta_{1k}|, |  \beta_{0k}|  )      \leq A(K-1),
\end{eqnarray*}
we have $A\geq \RD_{ED}/(K-1)$. The equality is attainable if and only if (c1) $\alpha_k=\RD_{ED}/(K-1)$, and $ \beta_{1k} = \beta_{0k} = 1$ for $k=1,\ldots,(K-1)$; or (c2) $\alpha_k=-1$, and $\beta_{1k} = \beta_{0k} = -1$ for $k=1,\ldots,K.$
The condition (c1) requires that the risk difference of the exposure $E$ on each category of $U$ to be the same as $\RD_{ED}/(K-1)$, and the confounder $U$ is a perfect predictor of the disease $D$. Similar interpretation applies to condition (c2).

Since 
\begin{eqnarray*}
\RD_{ED} &=&  \Big|  \sum_{k=1}^{K-1} \alpha_k \{  \beta_{1k} \pr(E=0) + \beta_{0k}\pr(E=1)  \} \Big| \\
&\leq &    \sum_{k=1}^{K-1}  | \alpha_k |   \max( | \beta_{1k}|, |  \beta_{0k}|  )    \leq B \sum_{k=1}^{K-1} |\alpha_k| \\
& \leq& B \sum_{k=1}^{K-1} \pr(U=k\mid E=1) + B  \sum_{k=1}^{K-1} \pr(U=k\mid E=0) \leq 2B,
\end{eqnarray*}
the lower bound for $B$ is $B\geq \RD_{ED}/2.$
The equality is attainable if and only if $\pr(U=0\mid E=0) = \pr(U=0\mid E=1)= 0, \pr(U=k\mid E=1)  \pr(U=k\mid E=0) = 0$ for $k=1,...,(K-1)$, and $\beta_{1k} = \beta_{0k} = \pm \RD_{ED}/2$ with the same sign as $\alpha_k$.

Since $\RD_{ED} \leq (K-1) AB \leq (K-1)\max^2(A,B)$, we have $\max(A,B)\geq \{ \RD_{ED}/(K-1) \}^{1/2}$, with the equality attainable if and only if $\alpha_k = \beta_{1k}  = \beta_{0k} = \pm \{ \RD_{ED}/(K-1) \}^{1/2}$ for $k=1,\ldots,K-1$. Due to the constraint $\sum_{k=1}^{K-1} | \alpha_k | \leq 2$ discussed above, the equality is attainable if and only if $(K-1)\{ \RD_{ED}/(K-1) \}^{1/2} \leq 2$ or $(K-1) \RD_{ED}\leq 4$. When $(K-1)\RD_{ED} > 4$, $B$ can attain its lower bound $\RD_{ED}$ with $\sum_{k=1}^{K-1} |\alpha_k| = 2.$ Therefore, $A$ can attain its lower bound $2/(K-1)$, which, in this case, is smaller than $\RD_{ED}/2.$
In summary, the lower bound for $\max(A,B)$ is $\max(A,B) \geq \{  \RD_{ED} / (K-1) \}^{1/2},$ if $(K-1)\RD_{ED} \leq  4$, and $\max(A,B) \geq \RD_{ED} / 2$, if $(K-1)\RD_{ED} > 4$. Equivalently, we have $\max(A,B) \geq \max\left[   \{\RD_{ED}/(K-1) \}^{1/2} , \RD_{ED}/2 \right] .$
\end{proof}

\begin{proof}[of Theorem 3]
The bound for $A$ remains the same.
Since 
\begin{eqnarray*}
\RD_{ED} &=& \Big|  \sum_{k=1}^{K-1} \alpha_k \{  \beta_{1k} \pr(E=0) + \beta_{0k}
(E=1)  \} \Big| \\
&\leq& B \sum_{k=1}^{K-1} |\alpha_k| \leq B (-\alpha_0) \leq B,
\end{eqnarray*}
the lower bound for $B$ is $B\geq \RD_{ED}$
The equality is attainable if and only if $\alpha_0 = -1$ and $\beta_{1k}= \beta_{0k} = \RD_{ED}$ for $k=1,\ldots,K-1$. The condition requires that the presence or absence of the confounder $U$ is perfectly predictive to the exposure $E$, and each category of $U$ is equally predictive to the disease $D$.

Since $\RD_{ED} \leq (K-1) AB \leq (K-1)\max^2(A,B)$, we have $\max(A,B)\geq \{ \RD_{ED}/(K-1) \}^{1/2}$, with the equality attainable if and only if $\alpha_k = \beta_{1k}  = \beta_{0k} = \pm \{ \RD_{ED}/(K-1) \}^{1/2}$ for $k=1,\ldots,K-1$. Due to the constraint $\sum_{k=1}^{K-1}  \alpha_k = -\alpha_0  \leq 1$ discussed above, the equality is attainable if and only if $(K-1)\{ \RD_{ED}/(K-1) \}^{1/2} \leq 1$ or $(K-1) \RD_{ED}\leq 1$. When $(K-1)\RD_{ED} > 1$, $B$ can attain its lower bound $\RD_{ED}$ with $\sum_{k=1}^{K-1} \alpha_k = 1.$ Therefore, $A$ can attain its lower bound $1/(K-1)$, which, in this case, is smaller than $\RD_{ED}.$
In summary, the lower bound for $\max(A,B)$ is $\max(A,B) \geq \{ \RD_{ED} / (K-1) \}^{1/2},$ if $(K-1)\RD_{ED} \leq  1$, and $\max(A,B) \geq \RD_{ED} $, if $(K-1)\RD_{ED} > 1$. Equivalently, we have $\max(A,B) \geq \max\left[    \{ \RD_{ED}/(K-1) \}^{1/2} , \RD_{ED} \right].$
\end{proof}

\end{document}